\documentclass[12pt,reqno,a4paper]{amsart}
\usepackage{amsmath,amssymb,dsfont,graphicx,cite,latexsym,
epsf,cancel,epic,eepic}
\usepackage[colorlinks=false]{hyperref}
\usepackage{mathpazo}
\usepackage{tikz-cd}

\newtheorem{theorem}{Theorem}

\newcommand{\bbR}{{\mathbb R}}
\newcommand{\bbC}{{\mathbb C}}

\def\t{\widetilde}

\title{How one can repair non-integrable Kahan discretizations}
\author{Matteo Petrera, Yuri B. Suris, Ren\'e Zander }

\thanks{E-mail: {\tt  petrera@math.tu-berlin.de, suris@math.tu-berlin.de, zander@math.tu-berlin.de}}

\begin{document}

\maketitle

\begin{center}
{\footnotesize{
Institut f\"ur Mathematik, MA 7-1\\
Technische Universit\"at Berlin, Str. des 17. Juni 136,
10623 Berlin, Germany
}}
\end{center}

\begin{abstract}
Kahan discretization is applicable to any system of ordinary differential equations on $\mathbb R^n$ with a quadratic vector field,
$\dot{x}=f(x)=Q(x)+Bx+c$, and produces a birational map $x\mapsto \widetilde{x}$ according to the formula $(\widetilde{x}-x)/\epsilon=Q(x,\widetilde{x})+B(x+\widetilde{x})/2+c$, where $Q(x,\widetilde{x})$ is the symmetric bilinear form corresponding to the quadratic form $Q$. When applied to integrable systems, Kahan discretization preserves integrability much more frequently than one would expect a priori, however not always. We show that in some cases where the original recipe fails to preserve integrability, one can adjust coefficients of the Kahan discretization to ensure its integrability. 
\end{abstract}

\section{Introduction}

The Kahan discretization was introduced in the unpublished notes \cite{K} as a method applicable to any system of ordinary differential equations on $\bbR^n$ with a quadratic vector field:
\begin{equation}\label{eq: diff eq gen}
\dot{x}=f(x)=Q(x)+Bx+c,
\end{equation}
where each component of $Q:\bbR^n\to\bbR^n$ is a quadratic form, while $B\in{\rm Mat}_{n\times n}(\bbR)$ and $c\in\bbR^n$. Kahan's discretization reads as
\begin{equation}\label{eq: Kahan gen}
\frac{\widetilde{x}-x}{\epsilon}=Q(x,\widetilde{x})+\frac{1}{2}B(x+\widetilde{x})+c,
\end{equation}
where
\[
Q(x,\widetilde{x})=\frac{1}{2}\big(Q(x+\widetilde{x})-Q(x)-Q(\widetilde{x})\big)
\]
is the symmetric bilinear form corresponding to the quadratic form $Q$. Equation (\ref{eq: Kahan gen}) is {\em linear} with respect to $\widetilde x$ and therefore defines a {\em rational} map $\widetilde{x}=\Phi_f(x,\epsilon)$. 
Clearly, this map approximates the time $\epsilon$ shift along the solutions of the original differential system. Since equation (\ref{eq: Kahan gen}) remains invariant under the interchange $x\leftrightarrow\widetilde{x}$ with the simultaneous sign inversion $\epsilon\mapsto-\epsilon$, one has the {\em reversibility} property
\begin{equation}\label{eq: reversible}
\Phi^{-1}_f(x,\epsilon)=\Phi_f(x,-\epsilon).
\end{equation}
In particular, the map $\Phi_f$ is {\em birational}. 
Kahan applied this discretization scheme to the famous Lotka-Volterra system and showed that in this case it possesses a very remarkable non-spiralling property. This property was explained by Sanz-Serna \cite{SS} by demonstrating that in this case the numerical method preserves an invariant Poisson structure of the original system.

The next intriguing appearance of this discretization was in two papers by Hirota and Kimura who (being apparently unaware of the work by Kahan) applied it to two famous {\em integrable} system of classical mechanics, the Euler top and the Lagrange top \cite{HK, KH}. Surprisingly, the discretization scheme produced in both cases {\em integrable} maps. 

Since then, integrability properties of the Kahan's method when applied to integrable systems (also called ``Hirota-Kimura method'' in this context) were extensively studied, mainly by two groups, in Berlin \cite{PS, PPS1, PPS2, PPS3, PS2, PS3, PS4, PZ, PSS, Zander} and in Australia and Norway \cite{CMOQ1, CMOQ2, CMOQ4, CMOQ5, KCMMOQ, KMQ}. It was demonstrated that, in an amazing number of cases, the method preserves integrability in the sense that the map $\Phi_f(x,\epsilon)$ possesses as many independent integrals of motion as the original system $\dot x=f(x)$. It was even conjectured  in \cite{PPS1} that this always would be the case, at least for algebraically integrable systems. However, it became clear soon that this conjecture is not true.

One can arrive at a simple counterexample in dimension $n=2$ as follows. In \cite{PPS2, PZ, CMOQ4} three classes of homogeneous quadratic Hamiltonian vector fields were considered, for which Kahan discretization is integrable, namely 
\begin{equation} \label{nahm intro}
\begin{pmatrix} \dot x\\ \dot y \end{pmatrix}
 =\frac{1}{c(x,y)}\begin{pmatrix} \partial H/\partial y  \\ -\partial H/\partial x \end{pmatrix},
\end{equation}
where
\begin{equation*}
H(x,y)=\ell_1^{\gamma_1}(x,y)\ell_2^{\gamma_2}(x,y)\ell_3^{\gamma_3}(x,y), \quad c(x,y)=\ell_1^{\gamma_1-1}(x,y)\ell_2^{\gamma_2-1}(x,y)\ell_3^{\gamma_3-1}(x,y),
\end{equation*}
with $\gamma_1, \gamma_2, \gamma_3\in\bbR\setminus\{0\}$, and $\ell_i(x,y)=a_ix+b_iy$
are linear forms. Integrability takes place for $(\gamma_1,\gamma_2,\gamma_3)=(1,1,1)$, $(1,1,2)$, and $(1,2,3)$. In all three cases, all integral curves of the system \eqref{nahm intro} are of genus one, and the same holds true for all irreducible invariant curves of Kahan's discretization. 

If $(\gamma_1,\gamma_2,\gamma_3)=(1,1,1)$, one is dealing with a homogeneous cubic Hamiltonian. As discovered in \cite{CMOQ1}, Kahan's discretization remains integrable for arbitrary (i.e., also for non-homogeneous) cubic Hamiltonians. 

Consider the case $(\gamma_1,\gamma_2,\gamma_3)=(1,1,2)$. By a linear projective change of coordinates $(x,y)\sim [x:y:1]\in \bbC P^2$, we can arrange $\ell_1(x,y)=y+\frac{1}{2}x$, $\ell_2(x,y)=y-\frac{1}{2}x$, $\ell_3(x,y)=x$. Thus, differential equations correspond to $c(x,y)=x$, $H(x,y)=x^2(y^2-\frac{1}{4}x^2)$, 
\begin{equation} \label{112 intro}
\renewcommand{\arraystretch}{1.2}
\left\{\begin{array}{l} \dot x = 2xy, \\ 
 \dot y = x^2-2y^2.
\end{array} \right.
\end{equation} 
Kahan's discretization of this system reads:
\begin{equation} \label{d 112 intro}
\renewcommand{\arraystretch}{1.2}
\left\{\begin{array}{l}
 \t x - x = \epsilon(\t x y+x\t y), \\ 
 \t y - y = \epsilon(x\t x-2y\t y),
\end{array} \right .
\end{equation}
and results in the following birational map: 
$$
\t x= \frac{x(1+3\epsilon y)}{1+\epsilon y-2\epsilon^2 y^2-\epsilon^2 x^2}, \quad 
\t y=\frac{y-\epsilon y^2+\epsilon x^2}{1+\epsilon y-2\epsilon^2 y^2-\epsilon^2 x^2}.
$$
This map is integrable and possesses the following integral of motion:
$$
H(x,y;\epsilon)=\frac{x^2(y^2-\frac{1}{4}x^2)}{\big(1+\epsilon(y+x)\big)\big(1+\epsilon(y-x)\big)\big(1-\epsilon(y+x)\big)\big(1-\epsilon(y-x)\big)}.
$$
All level sets of the integral are quartic curves with two double points at $(0,\pm 1/\epsilon)$, and the irreducible ones have genus 1. There are three reducible level sets: the one corresponding to vanishing of the numerator of $H(x,y;\epsilon)$ (consisting of four lines), the one corresponding to vanishing of the denominator of $H(x,y;\epsilon)$ (consisting of other four lines), as well as the following one:
$$
(1-\epsilon^2y^2)(1-\epsilon^2(2x^2+y^2))=0,
$$
consisting of two lines and an irreducible conic.

One can now attempt to generalize this construction for a non-homogeneous case,  say for $H(x,y)=x^2(y^2-\frac{1}{4}x^2-\frac{1}{2}b)$. Then differential equations \eqref{nahm intro} read
\begin{equation} \label{112 pert intro}
\renewcommand{\arraystretch}{1.2}
\left\{\begin{array}{l} 
 \dot x = 2xy, \\ 
 \dot y = b+x^2-2y^2,
\end{array} \right.
\end{equation} 
and still have the above mentioned property: all integral curves are of genus 1. However, Kahan's discretization of this system,
\begin{equation} \label{d 112 pert intro}
\renewcommand{\arraystretch}{1.2}
\left\{\begin{array}{l}
 (\t x - x)/\epsilon = \t x y+x\t y, \\ 
 (\t y - y)/\epsilon = b+x\t x -2y\t y,
\end{array} \right .
\end{equation}
is non-integrable. This can be shown by means of the singularity confinement criterium, by observing that for all three indeterminacy points of $\Phi_f^{-1}$, the orbits never land at an indeterminacy point of $\Phi_f$. Equivalently, the dynamical degree of $\Phi_f$ equals 2, that is, its algebraic entropy equals $\log 2$  (cf. \cite{Viallet, Diller}). All these statements are not that easy to prove rigorously, but the numeric evidence is very convincing. Thus, the map $\Phi_f: (x,y)\mapsto (\t x,\t y)$ defined by \eqref{d 112 pert intro} is a counterexample to integrability of Kahan discretizations for algebraically completely integrable quadratic vector fields. 

The goal of this note is to demonstrate how this can be remedied.

\section{First example}

\begin{theorem}
The Kahan-type map given by 
\begin{equation} \label{map 1}
\renewcommand{\arraystretch}{1.2}
\left\{\begin{array}{l}
 (\t x - x)/  \epsilon=\t x y+x\t y, \\ 
 (\t y - y)/ \epsilon= b+x\t x -(2-\epsilon^2b) y\t y,
\end{array} \right .
\end{equation}
is integrable, with an integral of motion
\begin{equation} \label{integral map 1}
H(x,y;\epsilon)=\frac{x^2\big((1-\frac{1}{2}\epsilon^2b)y^2-\frac{1}{4}(1-\epsilon^2b)x^2-\frac{1}{2}b\big)}{\big(1+\epsilon(y+x)\big)\big(1+\epsilon(y-x)\big)\big(1-\epsilon(y+x)\big)\big(1-\epsilon(y-x)\big)}.
\end{equation}
\end{theorem}

\begin{proof}
Consider the following map (a symmetric QRT root, cf. \cite{QRT, Dui}):
\begin{equation}\label{QRT root 1}
\widetilde{u}=v, \quad \widetilde{v}=\frac{\alpha uv-1}{u-\alpha v}.
\end{equation}
It is a birational map of $\bbC P^2$ (with non-homogeneous coordinates $(u,v)$ on the affine part $\bbC^2\subset\bbC P^2$), admitting an integral of motion
\begin{equation}\label{integral QRT root 1}
K(u,v)=\frac{\alpha(u^2+v^2-1)-uv}{(u^2-1)(v^2-1)}.
\end{equation}
We perform a linear projective change of variables in $\bbC P^2$, given in the non-homogeneous coordinates by
\begin{equation}\label{change 1}
u=\frac{1+y}{x}, \quad v=\frac{1-y}{x}.
\end{equation}
\smallskip

\noindent
Then the first equation of motion in \eqref{QRT root 1}, $\t u=v$, turns into
$$
\t x-x=\t x y+x \t y.
$$
The second equation of motion in \eqref{QRT root 1} can be re-written as a bilinear relation
$$
\frac{1}{2}(\t u v-u \t v)=\frac{1}{2}(\t u v+u \t v)+1-\alpha(u\t u+v\t v).
$$
Upon substitution \eqref{change 1}, this turns into
$$
\t y-y= (1-y \t y)+x \t x-2\alpha(1+y\t  y).
$$
So, we come to the system
$$
\renewcommand{\arraystretch}{1.2}
\left\{\begin{array}{l} 
\t x-x=x \t y+\t xy, \\ 
\t y-y=(1-2\alpha)+x\t x-(1+2\alpha)y\t y.
\end{array} \right.
$$
Scale $x\mapsto \epsilon x$, $y\mapsto \epsilon y$, and set $1-2\alpha=\epsilon^2 b$, so that $\alpha=(1-\epsilon^2 b)/2$ and \linebreak $1+2\alpha=2-\epsilon^2 b$. Then we arrive at the Kahan-type system \eqref{map 1}. Integral \eqref{integral map 1} is nothing but \eqref{integral QRT root 1} in the new coordinates.
\end{proof}

\section{Second example}

We can extend results of the previous section by adding one more inhomogeneous term in the Hamiltonian: $H(x,y)=x^2(y^2-\frac{1}{4}x^2-\frac{2}{3}cx-\frac{1}{2} b)$. Then system \eqref{nahm intro} reads:
\begin{equation} \label{112 pert 2}
\renewcommand{\arraystretch}{1.2}
\left\{\begin{array}{l} 
 \dot x = 2xy, \\ 
 \dot y = b+2cx+x^2-2y^2.
\end{array} \right.
\end{equation} 
This system still has the above mentioned property: all integral curves are of genus 1. Kahan's discretization of this system,
\begin{equation} \label{d 112 pert 2}
\renewcommand{\arraystretch}{1.2}
\left\{\begin{array}{l}
 (\t x - x)/\epsilon = \t x y+x\t y, \\ 
 (\t y - y)/ \epsilon= b+c(x+\t x)+x\t x-2y\t y,
\end{array} \right .
\end{equation}
is non-integrable, like in the previous case $c=0$. However, it can be repaired, as follows.
\begin{theorem}
The Kahan-type map given by 
\begin{equation}\label{map 2}
\renewcommand{\arraystretch}{1.2}
\left\{\begin{array}{l}
(\t x - x)/ \epsilon=x\t y+\t x y , \\ 
(\t y-y )/ \epsilon= b+c(1-\epsilon^2b)(x+\t x)+\big(1-\epsilon^2c^2(2-\epsilon^2b)\big)x \t x -(2-\epsilon^2b)y \t y,
\end{array}\right.
\end{equation}is integrable, with an integral of motion
\begin{equation} \label{integral map 2}
H(x,y;\epsilon)=\frac{x^2\big((1-\frac{1}{2}\epsilon^2b)y^2-\frac{1}{4}a_1x^2-\frac{2}{3}c_1x-\frac{1}{2}b\big)}{m_1(x,y)m_2(x,y)m_3(x,y)m_4(x,y)}.
\end{equation}
where
$$
a_1=1-\epsilon^2b-\tfrac{4}{3}\epsilon^2c^2p,\quad c_1=cp,\quad p=\frac{(1-\epsilon^2 b)(1-\frac{1}{2}\epsilon^2 b)}{1-\frac{1}{3}\epsilon^2b},
$$
\begin{eqnarray*}
m_1(x,y) & = & 1+\epsilon y+\epsilon(1-\epsilon c)x, \\
m_2(x,y) & = & 1+\epsilon y-\epsilon(1+\epsilon c)x, \\
m_3(x,y) & = & 1-\epsilon y+\epsilon(1-\epsilon c)x, \\
m_4(x,y) & = & 1-\epsilon y-\epsilon(1+\epsilon c)x.
\end{eqnarray*}
\end{theorem}

\begin{proof}
The most general symmetric QRT root which is a birational map of $\bbC P^2$ of $\deg=2$, reads:
\begin{equation}\label{QRT root 2}
\t u=v, \quad \t v=\frac{\alpha uv+\beta u-1}{u-\alpha v-\beta}.
\end{equation}
It admits an integral of motion
\begin{equation}\label{integral QRT root 2}
K(u,v)=\frac{\alpha(\alpha+1)(u^2+v^2-1)-(\alpha+1)uv+\beta(u+v)-\beta^2}{(u^2-1)(v^2-1)}.
\end{equation}
We perform a linear projective change of variables in $\bbC P^2$, given in the non-homogeneous coordinates by
\begin{equation}\label{change 2}
u=\frac{1+\beta x+y}{x}, \quad v=\frac{1+\beta x-y}{x}.
\end{equation}
\smallskip

\noindent
To transform equations of motion \eqref{QRT root 2} into new coordinates, it is useful to re-write the second one as a bilinear relation
$$
1+u \t v-\alpha u \t u-\alpha v\t v-\beta u-\beta\t v=0.
$$
Upon substitution \eqref{change 2} and some straightforward simplifications, we come to the following system:
\begin{equation*}
\renewcommand{\arraystretch}{1.2}
\left\{\begin{array}{l}
\t x-x=x\t y+\t xy, \\
\t y -y=(1-2\alpha)-2\alpha\beta(x+\t x)+\big(1-\beta^2(1+2\alpha)\big)x\t x-(1+2\alpha)y\t y.
\end{array}\right.
\end{equation*}
It remains to introduce a small parameter $\epsilon$ to make the above map to a discretization of a vector field. To this end, scale $x\mapsto \epsilon x$, $y\mapsto \epsilon y$, and set $1-2\alpha=\epsilon^2 b$, so that $1+2\alpha=2-\epsilon^2 b$, and $\beta=-\epsilon c$. Then we arrive at the Kahan-type system \eqref{map 2}. Integral \eqref{integral map 2} is the function \eqref{integral QRT root 2} expressed in the new coordinates.
\end{proof}

\section{Conclusions}

The main message of our examples is the following. The definition of Kahan's discretization, as used up to now, includes a very straightforward dependence on the small stepsize $\epsilon$. Namely, it only appears in the denominator of the differences $(\t x-x)/\epsilon$ which approximate the derivatives $\dot x$, compare \eqref{eq: diff eq gen} and \eqref{eq: Kahan gen}. On the contrary, coefficients of the bilinear expressions on the right hand sides of \eqref{eq: Kahan gen} are traditionally taken to literally coincide with the coefficients of the quadratic vector fields on the right hand sides of \eqref{eq: diff eq gen}. Even this straightforward discretization preserves integrability much more frequently than one would expect a priori. However, not always. What we show in this note, is the possibility to adjust the coefficients of the Kahan-type discretizations to ensure their integrability in some cases where the straightforward recipe fails to preserve integrability. It will be an important and entertaining task to find such integrable adjustments for other cases of non-integrability of the straightforward Kahan discretization. 

This research is supported by the DFG Collaborative Research Center TRR 109 ``Discretization in Geometry and Dynamics''.

{}
\end{document}